\documentclass[10pt,conference]{IEEEtran}
\newlength{\pic}
\usepackage{graphicx}
\graphicspath{ {images/} }
\usepackage{amsthm,amssymb,amsmath}
\usepackage{pgf,pgfarrows,pgfnodes}
\usepackage{tikz}
\usepackage{acronym}
\usepackage{algorithm, algorithmic}
\usepackage{lpic}
\usepackage{siunitx}
\usepackage{soul}
\usetikzlibrary{arrows,decorations.pathmorphing,backgrounds,positioning,fit,calc} 
\newtheorem{lemma}{Lemma}
\newtheorem{theorem}{Theorem}

\theoremstyle{definition}

\newtheorem{example}{Example}
\theoremstyle{remark}

\theoremstyle{plain}

\acrodef{LDPC}[LDPC]{low-density parity-check}
\acrodef{BEC}[BEC]{binary erasure channel}
\acrodef{BEWC}[BEWC]{binary erasure wiretap channel}
\acrodef{IoT}[IoT]{internet of things}
\acrodef{BER}[BER]{bit error rate}

\begin{document}


\title{Quantifying Equivocation for Finite Blocklength Wiretap Codes}

\author{\IEEEauthorblockN{Jack Pfister\IEEEauthorrefmark{1}, Marco A. C. Gomes\IEEEauthorrefmark{3}
Jo\~{a}o P. Vilela\IEEEauthorrefmark{4}, Matthieu R. Bloch\IEEEauthorrefmark{2}, and
Willie K. Harrison\IEEEauthorrefmark{1} \\
}

\IEEEauthorblockA{\IEEEauthorrefmark{1}Department of Electrical and Computer Engineering \\
University of Colorado Colorado Springs, Colorado Springs, CO 80923 \\ Email: jpfister@uccs.edu, wharriso@uccs.edu \\} 
\IEEEauthorblockA{\IEEEauthorrefmark{3}Instituto de Telecomunica\c{c}\~{o}es, Department of Electrical and Computer Engineering\\ University of Coimbra, Coimbra, Portugal \\ Email: marco@co.it.pt \\}
\IEEEauthorblockA{\IEEEauthorrefmark{4}CISUC and Department of Informatics Engineering, University of Coimbra, Coimbra, Portugal\\ Email: jpvilela@dei.uc.pt}
 \IEEEauthorblockA{\IEEEauthorrefmark{2}School of Electrical and Computer Engineering, Georgia of Institute of Technology, Atlanta, GA 30332\\ Email: matthieu.bloch@ece.gatech.edu}
}

%

\maketitle

\renewcommand{\thefootnote}{}
\footnotetext{This work was partially funded by the following entities and projects: the US National Science Foundation (Grant Award Number 1460085), project PTDC/EEI-TEL/3684/2014 (SWING2 - Securing Wireless Networks with Coding and Jamming) co-funded by COMPETE 2020 and Portugal 2020 - Programa Operacional Competitividade e Internacionaliza\c{c}\~{a}o (POCI), European Union through Fundo Europeu de Desenvolvimento Regional (FEDER) and Funda\c{c}\~{a}o para a Ci\^{e}ncia e Tecnologia (FCT), the FLAD project INCISE (Interference and Coding for Secrecy), and by Portuguese FCT under project UID/EEA/50008/2013. 
}
\renewcommand{\thefootnote}{\arabic{footnote}}

\begin{abstract}

This paper presents a new technique for providing the analysis and comparison of wiretap codes in the small blocklength regime over the binary erasure wiretap channel. A major result is the development of Monte Carlo strategies for quantifying a code's equivocation, which mirrors techniques used to analyze normal error correcting codes. For this paper, we limit our analysis to coset-based wiretap codes, and make several comparisons of different code families at small and medium blocklengths. Our results indicate that there are security advantages to using specific codes when using small to medium blocklengths.
\end{abstract}


\section{Introduction}
\label{sec:intro} 

Due to the increased number of automated and wireless devices in use today, it appears that the \ac{IoT} is slowly, but surely, becoming a reality. With the increased flexibility and convenience that the \ac{IoT} promises to bring about, come also a plethora of security and privacy issues. For one, the \ac{IoT} will be comprised of power-constrained devices; for two, these devices will likely need only short packets to communicate a large proportion of transmitted data; and for three, communications will need to have low latency to cope with small memory sizes on smaller connected devices~\cite{Durisi2016}. The architectures currently deployed in communication systems are unsuited for this new environment, as they typically rely on large blocklength coding schemes for reliability, including interleaving techniques that bring about added latency, and power-hungry and complicated algorithms for secret key exchange and/or cryptography. Thus, there is a current need for low-power secrecy algorithms that can make security guarantees over short blocklengths. 

One technique that may prove itself to be a nice match for many security and privacy issues in the \ac{IoT} is that of wiretap (or secrecy) coding~\cite{Harrison2013,Bloch2015} for physical-layer security~\cite{Wyner1975,BlochBook}. 
The general idea of such techniques is to code data in such a way that the channel over which an eavesdropper observes communciations naturally secures the data transmission, while also allowing reliability over other communications channels for legitimate receivers. 

If coding for secrecy is to prove itself adequate for solving the security issues inherent in the \ac{IoT}, it must be better understood how these codes perform in the finite blocklength regime, particularly with very short blocklengths. Traditionally, wiretap codes are evaluated and analyzed as blocklengths approach infinity using information theoretic security measures. 
Let a message $M$ be encoded into a length $n$ codeword $X^n$ for transmission across a communications channel. The eavesdropper observes a possibly noisy version of $X^n$ denoted by $Z^n$.  Data are transmitted with \emph{weak secrecy}~\cite{Wyner1975} if the leakage rate of information about the message goes to zero in the limit; that is,
\begin{equation}
\label{eq:weakSecrecy}
 \frac{1}{n}\mathbb{I}(M;Z^n) \rightarrow 0 \text{ as } n \rightarrow \infty.
 \end{equation}
Data are communicated with \emph{strong secrecy}~\cite{Maurer1994,Maurer2000} if the total amount of leaked information about the original message approaches zero as blocklength approaches infinity, or equivalently
\begin{equation}
\label{eq:strongSecrecy}
 \mathbb{I}(M;Z^n) \rightarrow 0 \text{ as } n \rightarrow \infty.
\end{equation}
While a majority of secrecy coding structures~(e.g., \cite{Harrison2013,Thangaraj2007,Subramanian2010}) make use of these measures to classify their security achievements, we argue that a new approach in the finite blocklength regime, beginning with extremely short blocklength codes, would be of great value. Furthermore, we wish to actually quantify the total equivocation as a function of channel parameters in the eavesdropper's channel, rather than only analyzing codes in the asymptotic blocklength regime.  
In this paper, we analyze coset-based secrecy codes (as originally presented in~\cite{Wyner1975,Thangaraj2007}) over finite blocklengths to quantify exactly (where possible) or estimate (using Monte Carlo techniques) the precise amount of information-theoretic security in terms of the equivocation
\begin{equation}
 \Delta = \mathbb{H}(M|Z).
\end{equation}
In essence, we are proposing that finite blocklength secrecy codes can be analyzed individually using simulation techniques similar to those that create \ac{BER} curves in generic error-control codes. In other words, when possible, we can give the full equivocation, or bound it as appropriate; but when these techniques fail, or when more precise security measures are required at specific blocklengths, we can simply estimate the equivocation using Monte Carlo simulations.

The remainder of this paper is organized as follows. Section \ref{sec:background} contains background information about the channel model used throughout the paper, coset coding in general, and a specific encoding and decoding algorithm. Section \ref{sec:security} demonstrates how to quantify equivocation when using coset coding techniques over binary erasure channels. Section \ref{sec:MonteCarlo} introduces a new parameter to compare finite-length codes with the achievable secrecy limits under the infinite blocklength assumption. Finally, Sections \ref{sec:characterizing} and \ref{sec:conclusion} present empirical results for different coset coding techniques and summarize the major findings of the paper, respectively.

\section{Background}
\label{sec:background}
In this section, we discuss the channel model used for this paper, as well as existing techniques for wiretap coding over the \ac{BEWC}.

\subsection{Overview of Channel Model}
The channel model assumed in this paper is a variant of Wyner's wiretap model~\cite{Wyner1975} called the \ac{BEWC}, which is depicted in Fig.~\ref{fig:wiretap}.
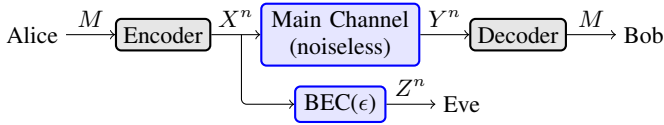
\begin{figure}
\small
\centering
  \begin{tikzpicture}
    [node distance=0.65cm, rounded corners=2pt, channel/.style={rectangle,draw=blue,fill=blue!10,thick,
    text centered,minimum size=4mm},
    boxedNode/.style={rectangle,draw,fill=black!10,thick,text centered, minimum size=4mm},
    inner sep=1mm]
    \node (Alice) {Alice};
    \node [boxedNode] (Encoder) [right=of Alice] {Encoder};
    \node [channel, text width=19mm] (Qm) [right =of Encoder] {Main Channel \\ (noiseless)};
    \node [boxedNode] (Decoder) [right=of Qm] {Decoder};
    \node (Bob) [right=of Decoder] {Bob};
    \begin{scope}[node distance=0.25cm]
      \node [channel, text width=10mm] (Qw) [below =of Qm] {BEC($\epsilon$)};
    \end{scope}
    \node (Eve) [right=of Qw] {Eve};
    \draw[->] (Alice) -- node [above] {$M$} (Encoder);
    \draw[->] (Encoder) -- node [above] {$X^n$} (Qm);
    \draw[->] (Qm) -- node [above] {$Y^n$} (Decoder);
    \draw[->] (Decoder) -- node [above] {$M$} (Bob);
    \draw[->] (Qw) to node [above] {$Z^n$} (Eve);
    \draw[->] ($(Encoder.east) + (4mm,0)$) |- (node cs:name=Qw,anchor=west);
  \end{tikzpicture}
\caption{The \ac{BEWC} model has a noiseless main channel between Alice and Bob, and a \ac{BEC} eavesdropper's channel.}\label{fig:wiretap}
\end{figure}
In this model, Alice wishes to securely transmit a binary message to Bob in the presence of Eve, an eavesdropper who has full knowledge of the coding scheme in use. Alice encodes a message $M$ from the alphabet $\mathcal{M} = \{1,2,\ldots,2^k\}$ into a corresponding $n$-bit codeword $X^n$ ($n \geq k$). Alice transmits $X^n$ to Bob through the main channel of communication, and Bob receives  $Y^n$ at the output of the channel. From this observation, Bob decodes and forms his estimate of the original message, denoted $M'$. For the purposes of this paper, the main channel is assumed to be noiseless, and thus, $Y^n = X^n$ and $M' = M$. An eavesdropper named Eve observes $Z^n$ through the eavesdropper's channel, which is a \ac{BEC} with parameter $\epsilon$.  Each bit of $X^n$ is erased by the channel with probability $\epsilon$ independent of all other bits, and erasures are denoted as `$?$' symbols. 

\subsection{Overview of Coset Coding}
In general, data should be encoded to minimize the probability of error for Bob and to restrict the amount of information intercepted by Eve. Since the main channel is noiseless, however, we need not worry about Bob. Secrecy over the eavesdropper's channel can be achieved through the coset coding procedure described in~\cite{Wyner1975,Ozarow1984}. Let the messages, $M\in\mathcal{M}$, be chosen uniformly at random. An ($n,n-k$) linear block code $C_1$ (also referred to as the base code) is chosen that contains $2^{n-k}$ $n$-bit codewords~\cite{MoonCoding}. From $C_1$, $2^k$ cosets ($C_1,C_2, C_3,\ldots,C_{2^k}$) can be obtained. These cosets can be formed by choosing an $n$-bit vector $A^n  \notin C_1$ and adding it to each codeword in $C_1$ using addition over GF(2). This process is repeated (ensuring now that $A^n$ is also not a codeword in another coset) until all $2^k$ cosets are obtained. By Lagrange's Theorem of cosets~\cite{MoonCoding}, each coset will contain $2^{n-k}$ $n$-bit binary vectors. Each message is then assigned to a unique coset, forming a codebook that contains every binary vector in the $n$-bit space. To encode $M$, a codeword is chosen at random from its corresponding coset and is transmitted as $X^n$. Since the main channel is noiseless, Bob simply has to find $Y^n$ in the codebook and map back to $M$. Eve also has access to the codebook and can obtain $M$ provided $Z^n$ allows her to rule out all but one coset. If $Z^n$ contains erasures, it is possible to achieve a measure of security further explored in Section \ref{sec:security}. The following example depicts the encoding process and the possible security benefits of coset coding.

\begin{example}
\label{ex:example1}
Let $k = 2$, and let the elements in $\mathcal{M} = \{1, 2, 3, 4\}$ be equally likely. We choose the base code $C_1$ to be the (4, 2) linear, block code containing the codewords $\{0000, 0110, 1001, 1111\}$. Cosets are formed, and each message is arbitrarily mapped to a corresponding coset resulting in the codebook seen in Table \ref{tab:table1}.
\begin{table}
\caption{Codebook structure outlined in example 1.}
\begin{center}
\begin{tabular}{|c|cccc|} 
\hline
$M$ & \multicolumn{4}{c|}{Codewords} \\
\hline
1 & 0000 & 0110 & 1001 & 1111 \\
\hline
2 & 0001 & 0111 & 1000 & 1110 \\
\hline
3 & 0010 & 0100 & 1011 & 1101 \\
\hline
4 & 0011 & 0101 & 1010 &  1100\\
\hline 
\end{tabular}
\end{center}
\label{tab:table1}
\end{table}
Suppose we wish to transmit $m$ = 3. A codeword from the third coset is chosen at random, for example 1011, and is transmitted as $x^n$. Since the main channel is noiseless, Bob receives $y^n = x^n$ and he can map $y^n = 1011$ back to the message $m = 3$. Suppose that Eve observes $z^n = 10??$. Since each coset contains a codeword consistent with $z^n$ Eve cannot rule out any cosets, and hence, $\mathbb{H}(M|Z^n = z^n) = 2$ bits, and Eve gains no information from the observation.
\end{example}

\subsection{Practical Encoding and Decoding Algorithm}
A computationally efficient method for encoding and decoding data in the coset coding scheme was developed in~\cite{Thangaraj2007}. The message $M$ is mapped to $k$-bits and is now denoted as $M^k\in\{0,1\}^k$. We first select an ($n,n-k$) linear block code for $C_1$ with generator matrix $G$ and parity check matrix $H$. The rows of $H$ are denoted $h_1,h_2,\ldots,h_k$. We now create $k$ linearly independent $n$-bit vectors ($q_1,q_2,\ldots,q_k$) that satisfy the following conditions
\begin{equation}
\label{eq:algorithm1}
q_i \notin C_1\text{ for all }0 \leq i \leq k,
\end{equation}
\begin{equation}
\label{eq:algorithm2}
q_ih_i^T = 1\text{ for all }0 \leq i \leq k,
\end{equation}
\begin{equation}
\label{eq:algorithm3}
q_ih_j^T = 0 \text{ for } 0 \leq i,j \leq k, i \neq j.
\end{equation}
The last two requirements ensure that the syndrome equals the message. We will now create a matrix, $G'$, whose rows are $q_1,q_2,\ldots,q_k$. We also generate a random $(n-k)$ bit vector $v^{(n-k)}$ for each transmission. The encoding procedure is a simple matrix multiplication and is represented as
\begin{equation}
x^n = 
\begin{bmatrix}
m^k & v^{n-k} 
\end{bmatrix}
\begin{bmatrix}
G' \\
G
\end{bmatrix}.
\end{equation}
Using this encoding procedure, $m^k$ determines the coset while the particular codeword within the coset is determined by $v^{n-k}$. If it is assumed that $y^n$ is received erasure-free, the receiver calculates the syndrome to obtain
\begin{align}
\label{eq:syndromeCalculation}
s^k & = y^nH^T \\
 & = x^nH^T \\ 
 & = \left[\begin{matrix} m^k & v^{n-k} \end{matrix}\right] \left[\begin{matrix} G' \\ G \end{matrix}\right]H^T \\
 & = m^k(G'H^T) + v^{n-k}(GH^T) \\
 & = m^k,
\end{align}
because (\ref{eq:algorithm2}) and (\ref{eq:algorithm3}) ensure that $G'H^T = I_k$ and $GH^T = 0$ by definition, where $I_k$ is the $k\times k$ identity matrix. The authors of~\cite{Thangaraj2007} further reduce the complications required by the encoder resulting in very efficient algorithms.

\section{Equivocation Calculation over the \ac{BEC}}
\label{sec:security}
Let us now calculate $\mathbb{H}(M|Z)$ for an eavesdropper  in our system. Note that~\cite{CoverAndThomas}
%
\begin{equation}
\label{eq:equivocation}
\begin{split}
\mathbb{H}(M|Z^n) &= \sum_{z^n \in \mathcal{Z}^n} p(z^n)\mathbb{H}(M|Z^n=z^n),\\
		&= \mathbb{E}\left[\mathbb{H}(M|Z^n=z_n)\right],
\end{split}
\end{equation}
\noindent
where $\mathcal{Z}^n = \{0, 1, ?\}$.

\par The expression $\mathbb{H}(M|Z^n=z^n)$  measures Eve's level of uncertainty regarding the message conditioned upon a particular observation $z_n$ from the eavesdropper's channel and is measured in units of bits. Our goal is to maximize Eve's equivocation using coset coding. The following theorem quantifies Eve's equivocation for a specific observation $z^n$ given the number of erasures, the placement of erasures, and the generator matrix of $C_1$. 

\begin{theorem}
Assume $M^k$ is chosen uniformly at random from $\{0,1\}^k$. Let the $(n,n-k)$ linear, block code $C_1$ be the base code to be used in the coset coding scheme. Let $G$, a binary $(n-k) \times n$ matrix, be the generator matrix for $C_1$. Consider an instance of an eavesdropper's observation $z^n\in\{0,1,?\}^n$. Let $\mu$ represent the number of unerasured positions in observation $z^n$ and let $G_\mu$ be a binary matrix with dimensions $(n-k) \times \mu$ whose columns correspond to the unerasured column indicies of $G$. Then,
\begin{equation}
\label{eq:proof1}
\mathbb{H}(M^k|Z^n=z^n) = k - \mu + \mathrm{rank}(G_\mu).
\end{equation}
\label{thm:theorem1}
\end{theorem}

\begin{proof}
If $G_\mu$ has rank $r$, then there exist $2^r$ ways to fill in the revealed positions within the codewords of $C_1$. Due to the properties of cosets, there are also $2^r$ ways to fill in the the revealed positions within the codewords of any and all solitary cosets. With this in mind, there exist $2^{n-k}/2^r = 2^{n-k-r}$ possible codewords in each possible coset. There must exist $2^{n-\mu}$ total codewords consistent with $z^n$, therefore, $2^{n-\mu} / (2^{n-k-r})$ cosets are consistent with $z^n$. Since all cosets are equally likely,
\begin{align}
\label{eq:proof2}
\mathbb{H}(M^k|Z^n=z^n) & =\log_2(2^{k-\mu+r}) \nonumber \\ & = k - \mu + \mathrm{rank}(G_\mu).
\end{align}

\end{proof}
It should be noted that this result is stronger than that given in Theorem 2 of~\cite{Thangaraj2007}, which was derived from results in \cite{Ozarow1984}, and a similar observation was made in~\cite{Wick2013}.

\begin{example}
\label{ex:example2}
The base code used in Example \ref{ex:example1} has the following generator matrix 
$$G = \begin{bmatrix} 1 & 0 & 0 & 1 \\ 0 & 1 & 1 & 0 \end{bmatrix}.$$
Let us assume that the eavesdropper observes $z_1^n = w??w$, where $w \in \{0,1\}$ and its actual value is irrelevant. Using Theorem \ref{thm:theorem1}, $\mathbb{H}(M^k|Z^n=z_1^n)$ = 1 bit. However, if the eavesdropper observes  $z_2^n = ww??$,  $\mathbb{H}(M^k|Z^n=z_2^n)$ = 2 bits. Notice that the codewords consistent with $z_1^n$ in Table~\ref{tab:table1} are contained in only two of the four cosets, leaking one bit of information, while the codewords consistent with $z_2^n$ are spread amongst all four cosets, leaking zero bits of information.
\end{example}

\section{Monte Carlo Channel Simulation Technique}
\label{sec:MonteCarlo}
It is true that the choice of $C_1$ plays an important role in the equivocation of Eve. Although for small codes $\mathbb{H}(M|Z^n)$ may be calculated exactly by cycling through all possible $z^n\in\mathcal{Z}^n$ in (\ref{eq:equivocation}) and using Theorem~\ref{thm:theorem1}, this becomes computationally infeasable as blocklength grows to even moderate lengths. To estimate the security performance of any base code in the coset coding scheme, a Monte Carlo simulation can be performed.

\subsection{Methodology}
Let $G$ be the binary generator matrix with dimensions $(n-k) \times n$ for $C_1$. Recall that the eavesdropper's channel has probability of erasure $\epsilon$. The equivocation of a particular observation can be calculated using Theorem~\ref{thm:theorem1}. This process is repeated for a predetermined number of iterations, resulting in an estimate of the average equivocation. 
\begin{lemma}
The expected value of 
\begin{equation}
\label{eq:averageEquivocation}
\hat{H} = \frac{1}{N}\sum\limits_{i=1}^{N}H(M|Z=z_i),
\end{equation}
where $N$ is the number of iterations in a Monte Carlo simulation, is the true equivocation. Therefore, $\hat{H}$ is an unbiased estimator of $\Delta = H(M|Z^n)$.
\end{lemma}
\begin{proof}
The expected value of $\hat{H}$ is
\begin{equation}
\label{eq:expectedValue}
\begin{split}
\mathbb{E}[\hat{H}] &=\mathbb{E}\left[\frac{1}{N}\sum\limits_{i=1}^{N}H(M|Z=z_i)\right]\\ 
		    &= \frac{1}{N}\sum\limits_{i=1}^{N}\mathbb{E}\left[H(M|Z=z_i)\right]\\
		    &=H(M|Z),
\end{split}
\end{equation}
where the final line in the proof comes from (\ref{eq:equivocation}).
\end{proof}
Using the estimator $\hat{H}$, the security performance of any coset code can be thoroughly characterized by simulating across a range of $\epsilon$ values. Similar types of Monte Carlo simulations have been used to characterize bit error rates (BER) of forward error correcting codes~\cite{MoonCoding,LinAndCostello}, so it should not surprise us that simulation can be used to evaluate wiretap codes.

\subsection{Achievability Gap}
Let $C_1$ be the $(n,n-k)$ linear, block code that will be used as the base code in the coset coding scheme as before. In the worst case scenario where $z^n$ contains zero erasures, all the information is leaked to the eavesdropper, and $\mathbb{H}(M|Z^n) = 0$. In the best case scenario, $z^n$ contains sufficient erasures such that $\mathbb{H}(M|Z^n) = k = \mathbb{H}(M)$. It now makes sense to present equivocation on a normalized scale, and we note that 
\begin{equation}
 \frac{\Delta}{n} = \frac{\mathbb{H}(M|Z^n)}{n}
\end{equation}
is usually called the equivocation rate. 
Notice that this quantity can be bounded as
\begin{equation}
\label{eq:twoBoundEquivocation}
0 \leq \frac{\mathbb{H}(M|Z^n)}{n} \leq \frac{\mathbb{H}(M)}{n} = R,
\end{equation}
where $R = k/n$ is called the secret information rate, using the standard inequality rule of conditional entropy~\cite{CoverAndThomas}. Further note that the secrecy capacity $C_s$, defined as the supremum of rates such that weak or strong secrecy can be achieved while also maintaining reliable communications over the main channel, is equal to $\epsilon$ for the \ac{BEWC}~\cite{Thangaraj2007}. Thus, it is also true that
\begin{equation}
  \label{eq:epsilonBd}
  \frac{\mathbb{H}(M|Z^n)}{n} \leq \epsilon.
\end{equation}
Combining (\ref{eq:twoBoundEquivocation}) and (\ref{eq:epsilonBd}) results in the overall bound of
\begin{equation}
\label{eq:combinedBd}
  0 \leq \frac{\mathbb{H}(M|Z^n)}{n} \leq \min(\epsilon, R),
\end{equation}
which is depicted in Fig.~\ref{fig:Ag}.
Since for finite values of $n$, 
\begin{equation}
\label{eq:upperBoundEquivocation}
\frac{\mathbb{H}(M|Z^n)}{n} \Big|_{\epsilon=R} < R,
\end{equation}
but ideal secrecy codes can certainly do no better than $R$ in the limit when $\epsilon = R$, we can effectively judge how closely a finite blocklength code gets to approaching the asymptotic secrecy supremum by considering the gap between $\mathbb{H}(M|Z^n)$ and $R$ at $\epsilon = R$. Thus,
we now define the \emph{achievability gap}, $A_g$, as
\begin{equation}
\label{eq:achievabilityGap}
A_g = R - \frac{\mathbb{H}(M|Z^n)}{n} \Big|_{\epsilon=R}.
\end{equation}
Using Monte Carlo techniques, individual choices of $C_1$ in a coset coding scheme can now be compared side by side using their entire equivocation rate curves, or using a single metric $A_g$. Both of these are depicted in Fig.~\ref{fig:Ag}. As $A_g$ gets smaller, the equivocation rate curve also approaches the bound in (\ref{eq:combinedBd}), which is best possible, even for infinite length codes.
Therefore, good secrecy codes and codes with larger blocklengths will tend to have smaller $A_g$ values. The achievability gap is significant because it is the largest difference between the equivocation rate bound in (\ref{eq:combinedBd}) and a code's true equivocation rate. A code's equivocation rate is always a concave function of $\epsilon$~\cite{CoverAndThomas}. For $0 \leq \epsilon \leq R$ the bound in (\ref{eq:combinedBd}) is a linear function of $\epsilon$. Therefore, the difference between the bound and the code's true equivocation rate will continue to grow along this interval. Along the interval $R \leq \epsilon \leq 1$ the bound in (\ref{eq:combinedBd}) is a horizontal line. As are result, the difference between the bound and the code's true equivocation rate will shrink along this interval. Logically, the largest difference between the bound and the code's true equivocation must occur at $\epsilon = R$, precisely where the achievability gap is evaluated. The concept of the achievability gap is best understood with the following example.
\begin{example}
\label{ex:example3}
$C_1$ is chosen to be the (7,4) Hamming code with secret information rate $R \approx 0.4286$. Since the blocklength is reasonably small, the equivocation rate can be calculated exactly, as can all the equivocation rate curves for every linear block code with $n = 7$ and $k = 4$. 
The results of this experiment can be seen in Fig.~\ref{fig:Hamming_example}. Notice that $\frac{1}{n}H(M|Z^n) \Big|_{\epsilon=0.4286} < R$, as expected. For this code, $A_g \approx 0.0812$ bits. By inspection, it is easy to see that the largest difference between the bound and the true equivocation rate occurs at $\epsilon = R$. Careful inspection of the figure reveals that $A_g$ is actually minimized for (7,4) linear block codes in the choice of $C_1$ as the Hamming code.

\begin{figure}
	\centering
	\includegraphics[width=\columnwidth]{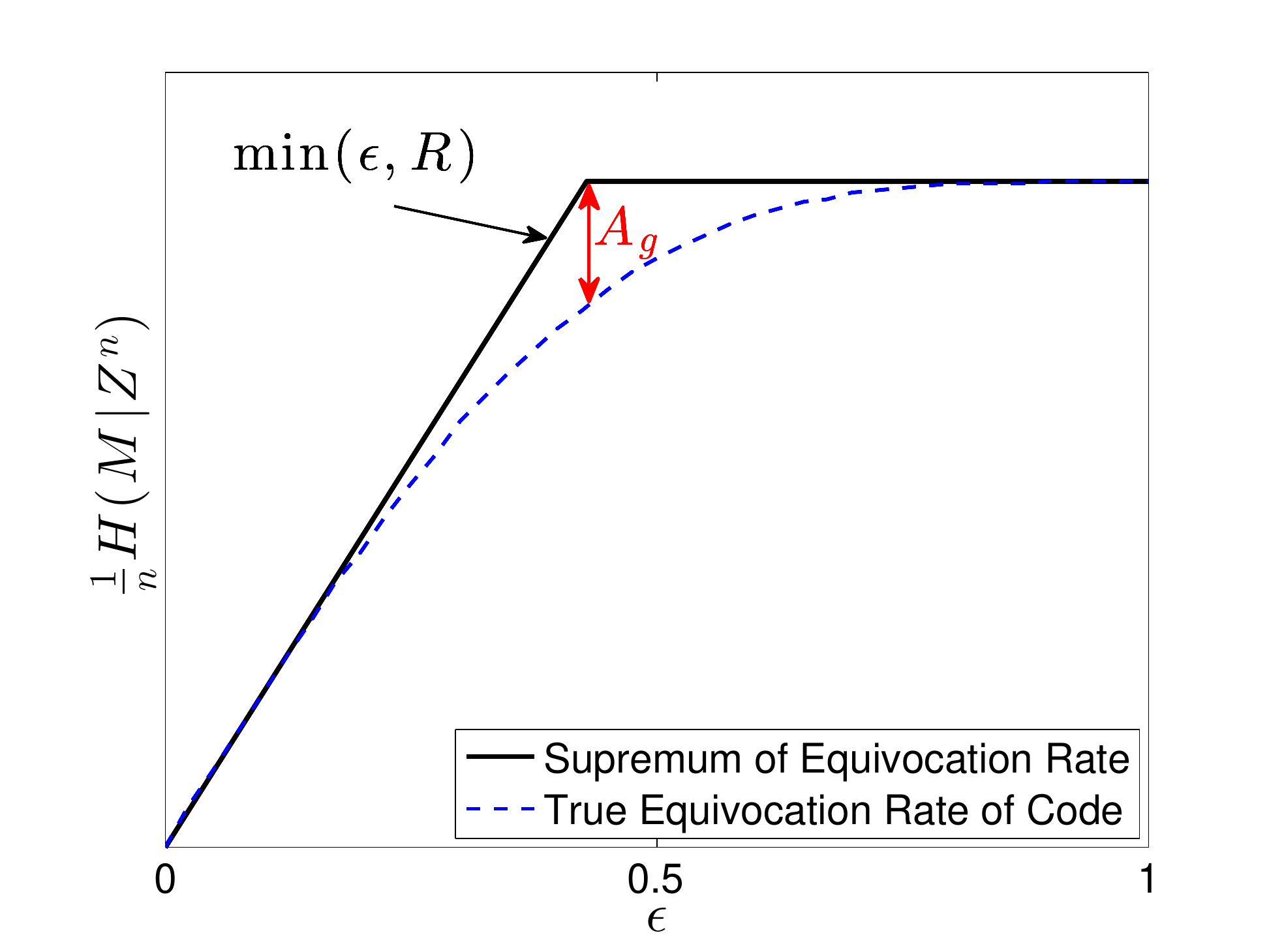}
	\caption{The achievability gap is the difference between the true equivocation rate of a code and it's supremum in the limit evaluated at $\epsilon = R$.}
	\label{fig:Ag}
\end{figure}
\end{example}

\section{Characterizing Algebraic and Random Codes with Small to Medium Blocklengths}
\label{sec:characterizing}
In this section, we present recommendations on how to characterize the security performance of small to medium blocklength codes. Ideally, the true equivocation rate of a code should be calculated through (\ref{eq:equivocation}). From a computational resource standpoint, this is only feasible for codes with blocklengths less than 10. For slightly larger codes, the logical next step would be to place bounds on the true equivocation rate. Steps toward bounding the true equivocation rate are presenented in~\ref{Wong2011}. We argue that for codes with blocklengths larger that 10, performing a Monte Carlo simulation described in Section \ref{sec:MonteCarlo} is a valid method to estimate the security performance of a code. However, due to the rank calculation in (\ref{eq:proof1}), the Monte Carlo simulation method is computationally expensive and is not feasible for codes with large blocklengths. Finally, we believe calculating the achievability gap for codes with large blocklengths gives some insight into the security performance of a code since it measures the maximum difference between the code's theoretical maximum equivocation rate and its estimated equivocation rate. We explore some of these ideas in the following subsections.

\subsection{Calculating True Equivocation for Small Blocklengths}
Example \ref{ex:example3} from the last section may cause us to wonder whether Hamming codes are, in fact, the best possible secrecy codes for their size parameters. By directly calculting the equivocation rate (\ref{eq:equivocation}), we have observed that Hamming and simplex codes are the best performing codes for their respective information rates. Figure~\ref{fig:Hamming_example} shows the equivocation rate curves for every (7,4) linear block code in a coset coding scheme, while Fig.~\ref{fig:Simplex_example} shows the curves for every (7,3) linear block code. We note that the Hamming code wins among the (7,4) codes, and its dual, the simplex code, wins among the (7,3) codes for every value of $\epsilon$. We also note in both figures that some codes perform better than their counterparts at larger values of $\epsilon$ but perform worse than their counterparts at smaller values of $\epsilon$ and vice versa. This makes it difficult to rank the codes in relation to one another (with the exception of the Hamming and simplex codes). Noting that these algebraic structures are quite interesting in a secrecy coding context, in the next section we investigate larger Hamming and simplex codes, and compare their equivocation rate curves and achievability gaps to those of randomly generated codes. 
\begin{figure}
	\centering
	\includegraphics[width=\columnwidth]{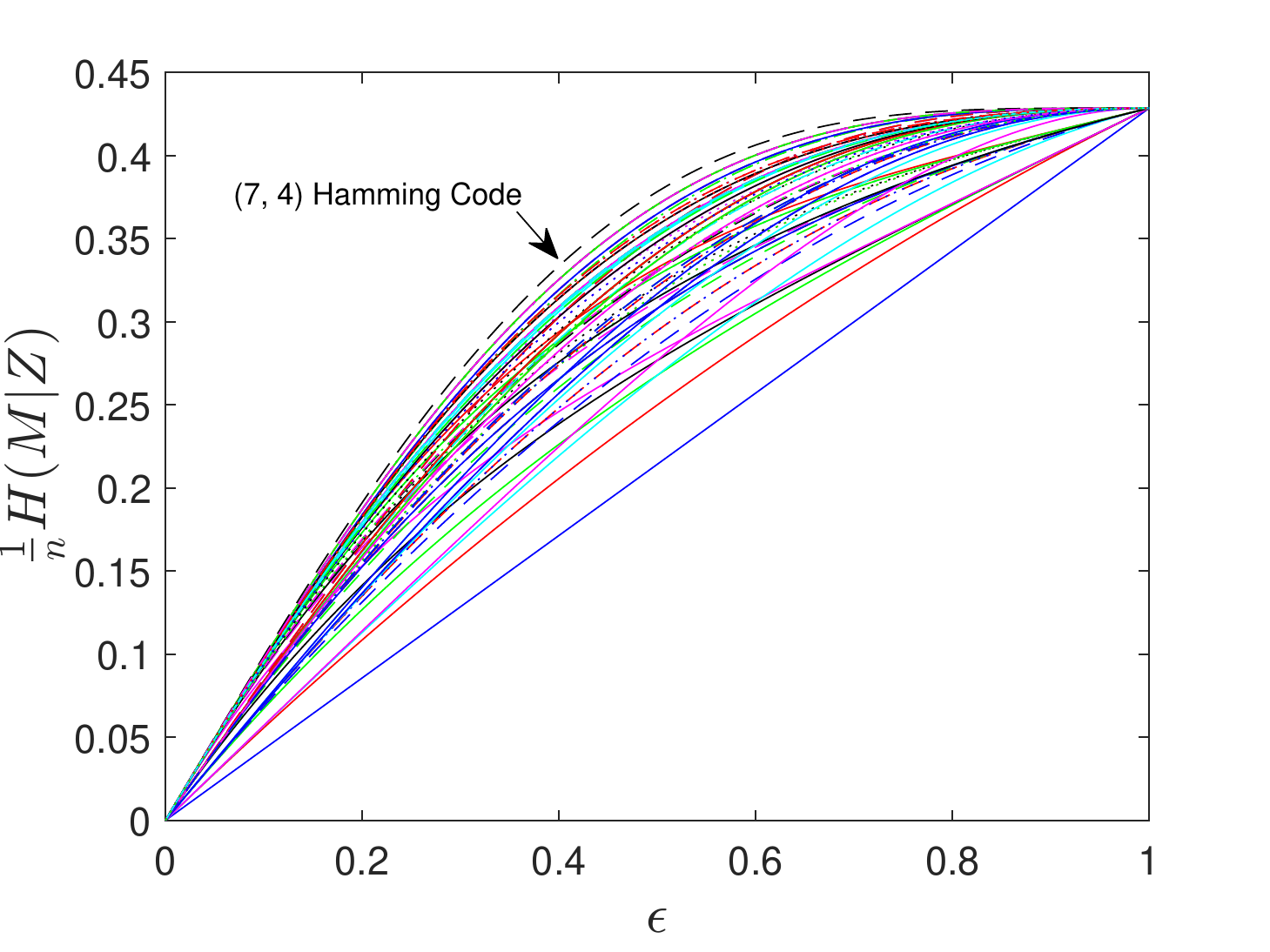}
	\caption{Equivocation rate curves for all (7,4) linear block codes. The Hamming code achieves the highest equivocation for all codes of this size.}
	\label{fig:Hamming_example}
\end{figure}
\begin{figure}
	\centering
	\includegraphics[width=\columnwidth]{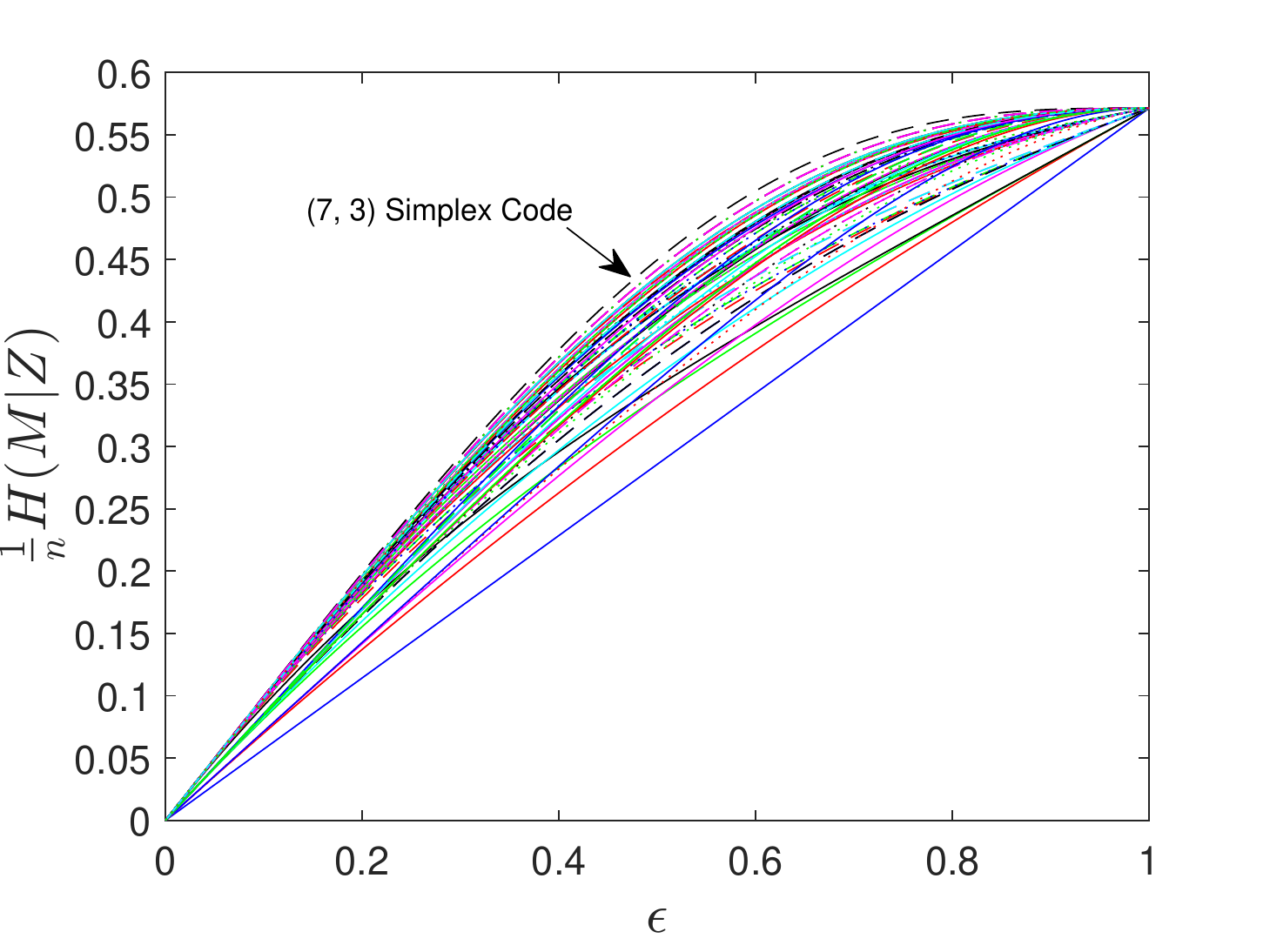}
	\caption{Equivocation rate curves for all (7,3) linear block codes. The simplex code achieves the highest equivocation for all codes of this size.}
	\label{fig:Simplex_example}
\end{figure}

\subsection{Estimating Equivocation Rates for Small to Medium Blocklengths}
Using the Monte Carlo simulation technique described earlier, experimental values of $A_g$ for Hamming and simplex codes with larger blocklengths were obtained and are given in Tables \ref{tab:table2} and \ref{tab:table3}. Here we note a general trend that the achievability gap $A_g$ shrinks as blocklength grows. This makes sense, because $A_g$ measures the difference between a code's equivocation rate and the supremum of achievable equivocation rates, which is to be understood in the limit as $n\rightarrow\infty$. Thus, larger codes should do better in general.
\begin{table}
\caption{Achievability gaps for hamming codes.}
\begin{center}
\begin{tabular}{ | c | c | c | c |} 
\hline
Blocklength & $R$ & $A_g$ (bits)\\
\hline
7 & 0.4286 & 0.0812\\
\hline
15 & 0.2667 & 0.0723\\
\hline
31 & 0.1613 & 0.0311\\
\hline
63 & 0.0952 & 0.0181\\
\hline
\end{tabular}
\label{tab:table2}
\end{center}
\end{table}
\begin{table}
\caption{Achievability gaps for simplex codes.}
\begin{center}
\begin{tabular}{ | c | c | c |} 
\hline
Blocklength & $R$ & $A_g$ (bits)\\
\hline
7 & 0.5714 & 0.0779\\
\hline
15 & 0.7333 & 0.0526\\
\hline
31 & 0.8387 & 0.0305\\
\hline
63 & 0.9048& 0.0179\\
\hline 
\end{tabular}
\label{tab:table3}
\end{center}
\end{table}
The full equivocation rate curves for the codes from each of these tables are given for both Hamming and simplex codes in Figs. \ref{fig:hammingEquivocations} and \ref{fig:simplexEquivocations}, respectively.
\begin{figure}
	\centering
	\includegraphics[width=\columnwidth]{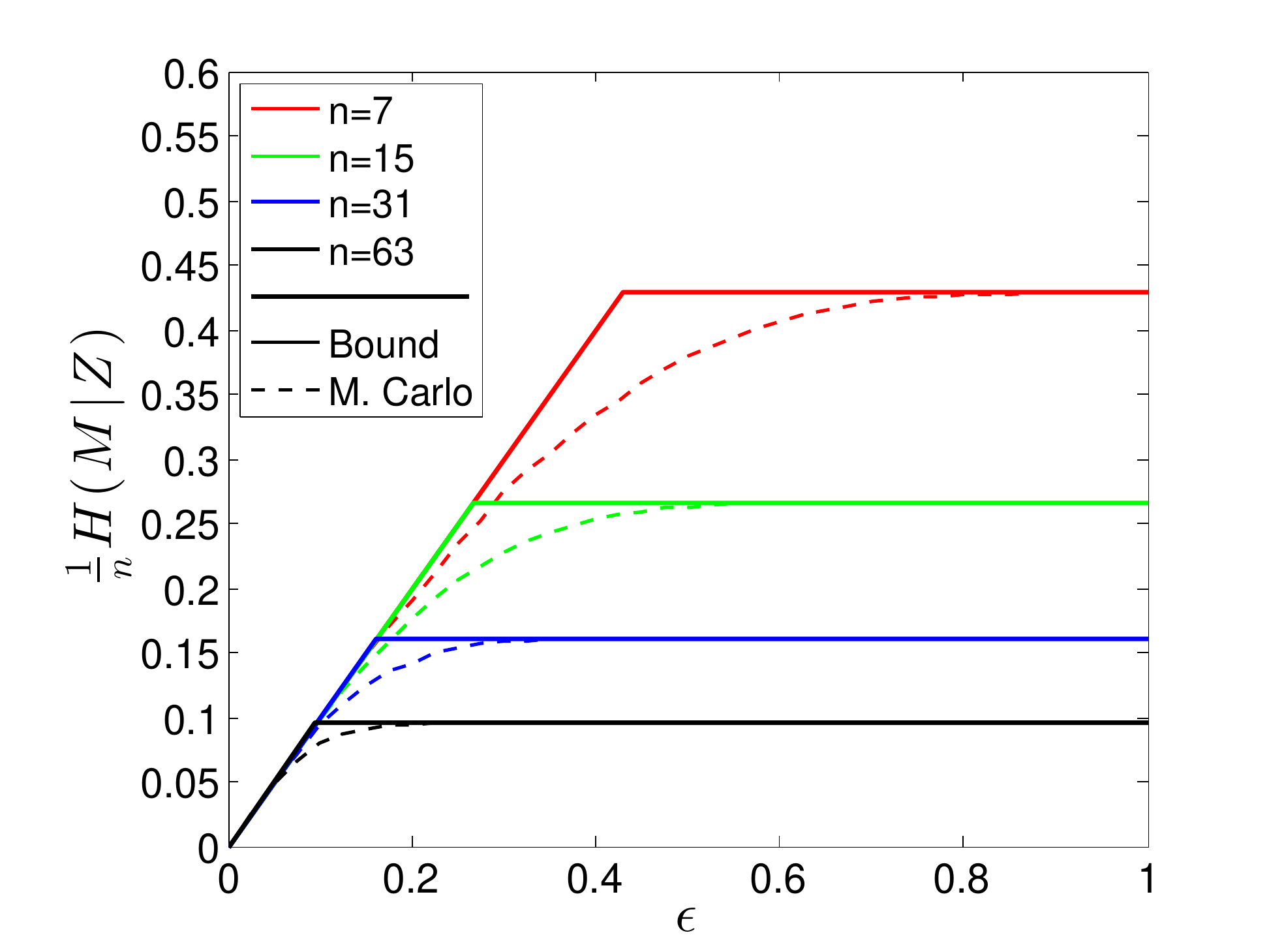}
	\caption{Equivocation rate curves for Hamming codes of four different blocklengths with upper bounds for each.}
	\label{fig:hammingEquivocations}
\end{figure}
\begin{figure}
	\centering
	\includegraphics[width=\columnwidth]{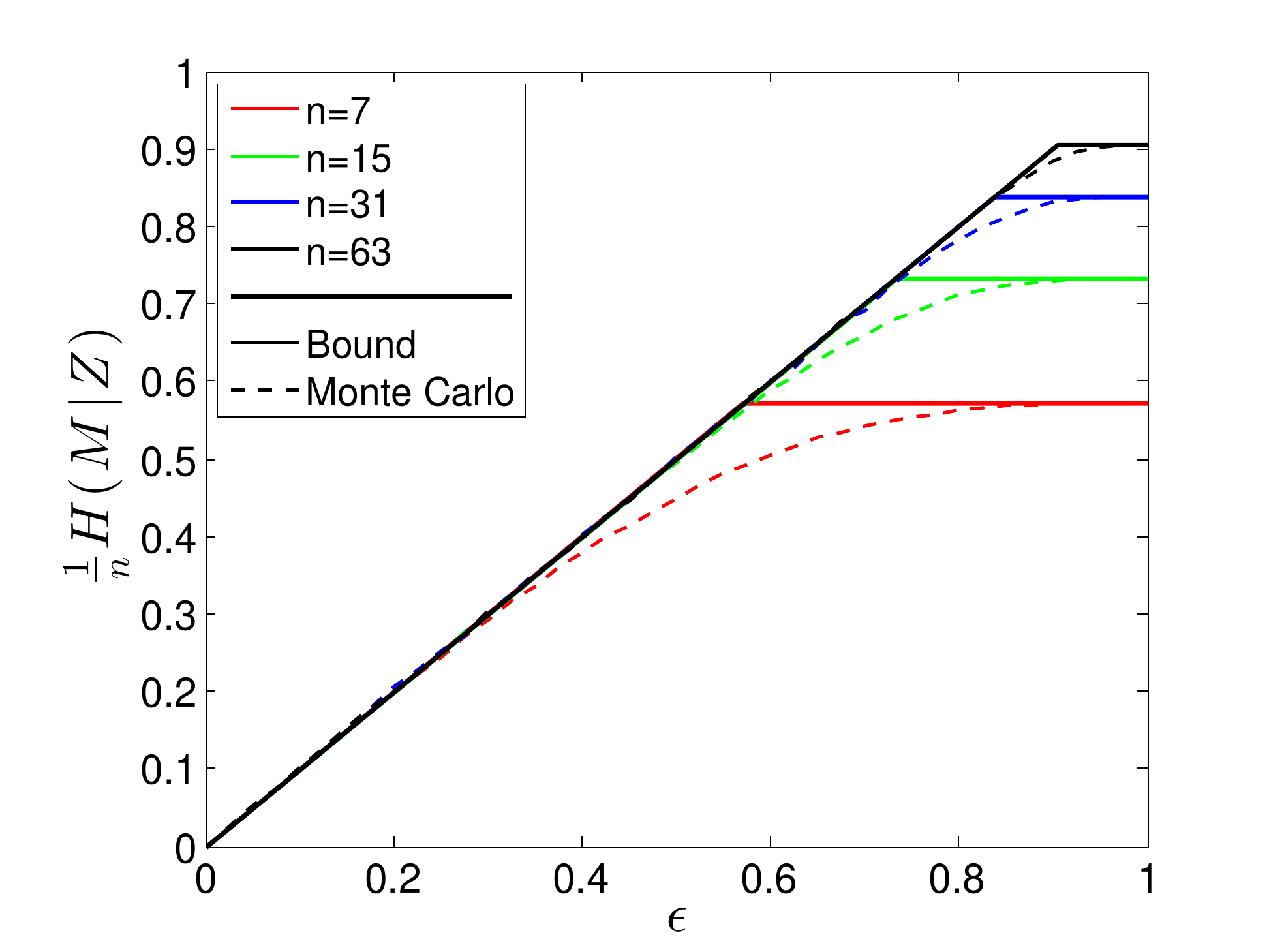}
	\caption{Equivocation rate curves for simplex codes of four different blocklengths with upper bounds for each.}
	\label{fig:simplexEquivocations}
\end{figure}

We now increase the blocklength and generate codes randomly so as to compare with these highly structured algebraic codes. 
The random codes that we consider have a single parameter $\alpha$, and generators for these codes are constructed such that each bit in the generator matrix is equal to one with probability $\alpha$, independent from all other bits. 
For blocklengths slightly larger than 10, we are no longer capable of calculating equivocation exactly in any reasonable amount of time. Thus, we employ the Monte Carlo techniques developed herein, and find that simulations show random codes with $\alpha \approx 0.5$, tend to have smaller $A_g$ values. Simulations also show that the (31,26) Hamming code slightly outperforms (31,26) random codes with $\alpha \approx 0.5$. To test this, ten (31,26) random codes with $\alpha \approx 0.5$ were created and tested using the Monte Carlo simulation techniques. The average security performance of these random codes compared to the (31,26) Hamming code's performance is shown in Fig. \ref{fig:randomHamming} with 95\% confidence intervals.
\begin{figure}
	\centering	 
	\includegraphics[width=\columnwidth]{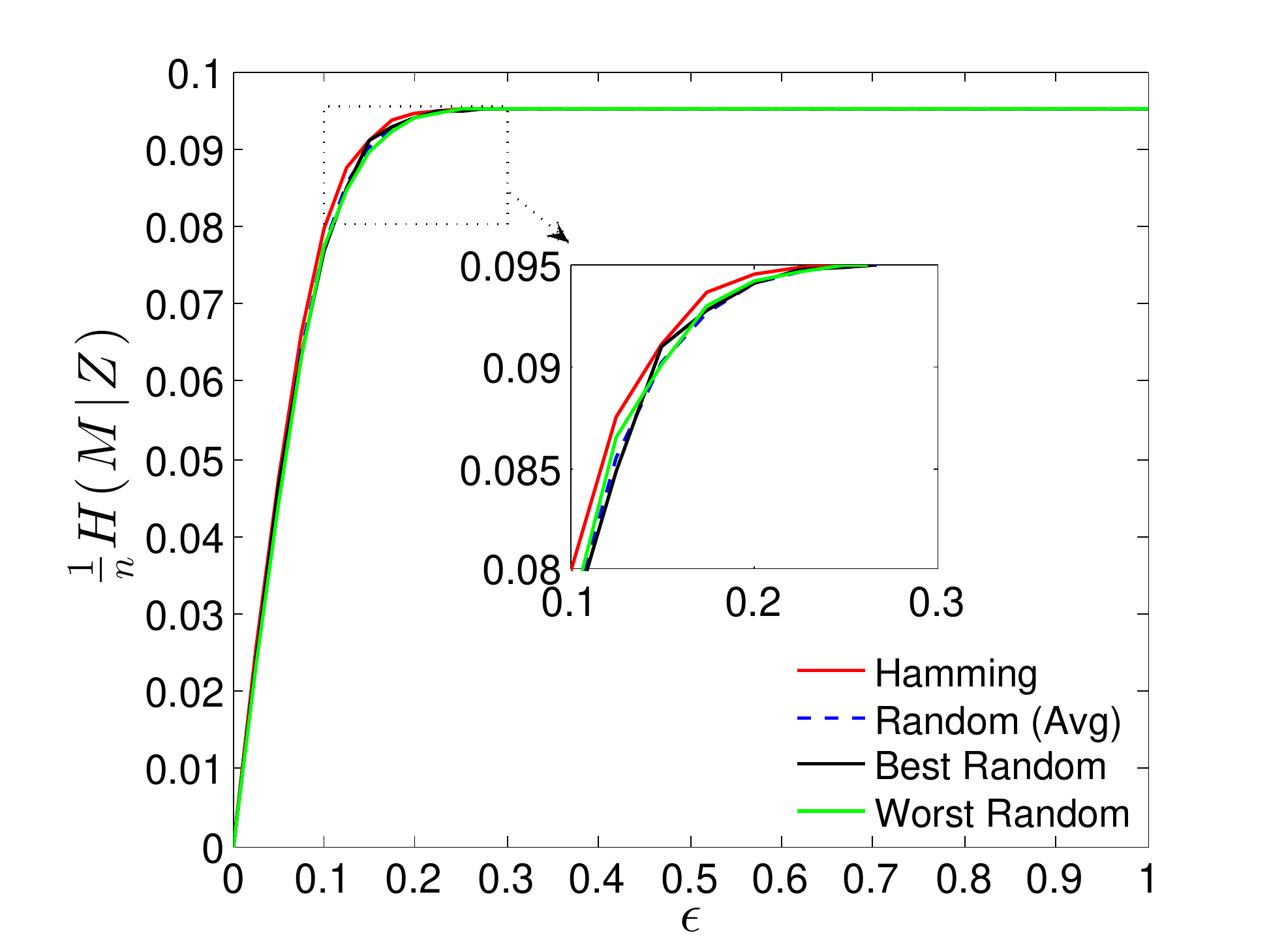}
	\caption{Equivocation rate curves for (63,57) random codes (average case, best case, and worst case) and the (63,57) Hamming code.}
	\label{fig:randomHamming}
\end{figure}
Simulations further indicate that the (31,5) simplex code outperforms (31,5) random codes with $\alpha \approx 0.5$. Again, ten (31,5) random codes with $\alpha \approx 0.5$ were created and tested using the Monte Carlo simulation techniques, and the results can be viewed in Fig.~\ref{fig:randomSimplex}.
\begin{figure}
	\centering
	\includegraphics[width=\columnwidth]{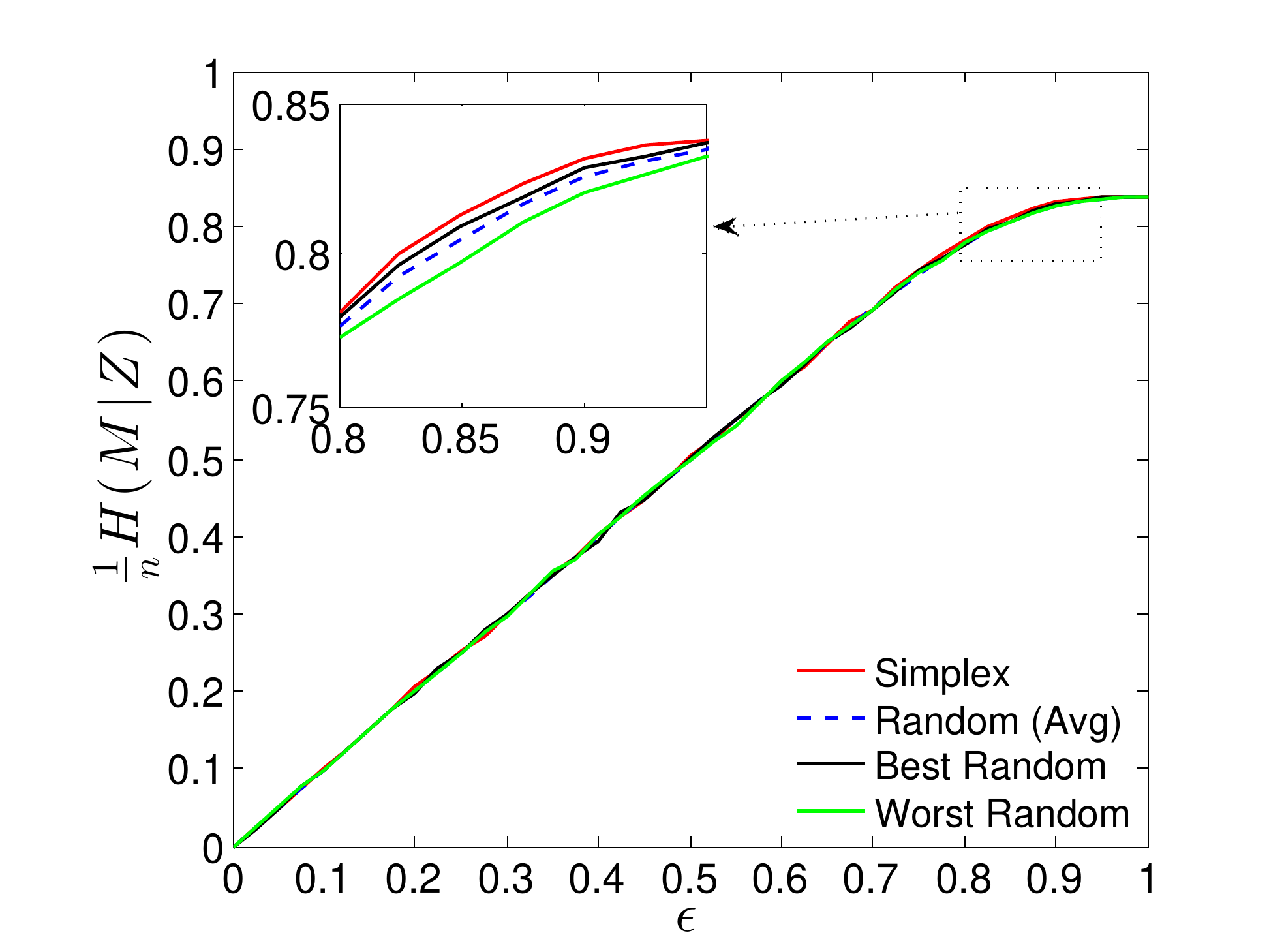}
	\caption{Equivocation rate curves for (31,5) random codes (average case, best case, and worst case) and the (31,5) simplex code.}
	\label{fig:randomSimplex}
\end{figure}
We expect the difference between algebraic codes and randomly chosen codes to further shrink for yet larger blocklengths, indicating that codes generated somewhat randomly may be expected to perform within some small difference to more optimized structures.

\section{Conclusion}
\label{sec:conclusion}
In conclusion, we have presented the idea of simulating equivocation rate curves using Monte Carlo techniques for secrecy code performance comparison in the finite blocklength regime, just as is commonly used to compare varying codes and code ensembles for general error-control codes. We have likewise presented a new parameter called the achievability gap that compares the equivocation rate curve to the optimal equivocation rate only achievable in the asymptotic blocklength regime. Small achievability gaps are preferable to larger ones in real secrecy code designs, and we presented some results for small blocklengths that indicated Hamming and simplex codes may have optimal structures for secrecy. However, as blocklength increased to even moderate sizes, the differences between these codes and randomly generated ones was small. However, as the \ac{IoT} gradually requires us to develop new, lightweight, and optimal security algorithms for small packet sizes, finding best possible codes may still be valuable.

\bibliographystyle{ieeetran}
\bibliography{references}

\end{document}